\documentclass[aps,prl,showpacs,english,10pt,a4,nofootinbib,notitlepage,twocolumn,superscriptaddress]{revtex4-1}
 
\usepackage{graphicx}
\usepackage{hyperref}
\usepackage{color}
\usepackage{amsthm,amsmath,amssymb,dsfont}
\usepackage{pgfplots}
\newtheorem{prop}{Proposition}

\newcommand{\ket}[1]{\vert #1 \rangle}
\newcommand{\bra}[1]{\langle #1 \vert}

\newcommand{\tr}{\mathrm{Tr}}
\newcommand{\ct}{^{\dagger}}
\newcommand{\av}[1]{\lvert #1\rvert}

\newlength\figureheight
\newlength\figurewidth

\begin{document}

\title{Robust Characterization of Loss Rates}
\author{Joel J. Wallman}
\author{Marie Barnhill}
\affiliation{Institute for Quantum Computing, University of Waterloo, 
Waterloo, Ontario N2L 3G1, Canada}
\affiliation{Department of Applied Mathematics, University of Waterloo, 
Waterloo, Ontario N2L 3G1, Canada}
\author{Joseph Emerson}
\affiliation{Institute for Quantum Computing, University of Waterloo, 
Waterloo, Ontario N2L 3G1, Canada}
\affiliation{Department of Applied Mathematics, University of Waterloo, 
Waterloo, Ontario N2L 3G1, Canada}
\affiliation{Canadian Institute for Advanced Research, Toronto, Ontario M5G 
1Z8, Canada}
\date{\today}

\begin{abstract}
Many physical implementations of qubits---including ion traps, optical lattices 
and linear optics---suffer from loss. A nonzero probability of irretrievably 
losing a qubit can be a substantial obstacle to fault-tolerant methods of 
processing quantum information, requiring new techniques to safeguard against 
loss that introduce an additional overhead that depends upon the loss rate. 
Here we present a scalable and platform-independent protocol for 
estimating the average loss rate (averaged over all input states) resulting 
from an arbitrary Markovian noise process, as well as an independent estimate 
of detector efficiency. Moreover, we show that our protocol gives an 
additional constraint on estimated parameters from randomized benchmarking that 
improves the reliability of the estimated error rate and provides a new 
indicator for non-Markovian signatures in the experimental data. We also derive 
a bound for the state-dependent loss rate in terms of the average loss rate. 
\end{abstract}


\maketitle

In order to build practical devices for processing and transmitting quantum 
information, the rate of decoherence and other errors must be below certain 
fault-tolerant thresholds. In particular, many experimental implementations of 
qubits---such as ion traps~\cite{Blakestad2011,Wright2013}, optical 
lattices~\cite{Vala2005} and linear optics~\cite{Fortescue2014}---suffer 
from irretrievable loss, that is, there is a nonzero probability of the qubit 
vanishing (as opposed to leaking to other energy levels). Such loss of 
normalization can be a substantial obstacle to many quantum information 
protocols, requiring different error-correction techniques to achieve 
fault-tolerance~\cite{Varnava2006,Fortescue2014,Whiteside2014}.
For example, the surface code may not be used directly if there is any 
probability of losing a qubit, while for the topological cluster states, 
loss rates of less than $1\%$ are required to avoid impractical 
overheads~\cite{Whiteside2014}.

However, there are two substantial challenges in characterizing loss. Firstly, 
the loss rate may depend on the state of the qubit, such as when a qubit is 
encoded in a superposition of vacuum and single-photon states. Secondly, the 
loss due to imperfect operations has to be distinguished from the inefficiency 
of the detector~\cite{Hogg2014}. Quantum process 
tomography~\cite{Chuang1997,Poyatos1997} could be used to characterize loss, 
however, it is inefficient in the number of qubits and is sensitive to state 
preparation and measurement (SPAM) errors~\cite{Merkel2013} and so cannot 
distinguish between loss due to imperfect operations and inefficient detectors.

In this Letter, we present a robust and efficient protocol that characterizes 
the loss rate due to imperfect operations averaged over input states. Our 
protocol is platform-independent, simple to implement and analyze and only 
assumes that the noise is Markovian. We begin by defining survival rates and 
then present our protocol and derive the associated analytical decay curve 
under the assumption of Markovian noise. We then prove that the average loss 
rate estimated via our protocol provides a practical bound on the loss rate for 
any state. Since our protocol is robust to SPAM, the choice of state and 
measurement in our protocol is unconstrained. However, we discuss two 
particularly suitable choices. The first of these maximizes the signal and the 
second allows one of the parameters in randomized 
benchmarking~\cite{Magesan2011} to be independently estimated and leads to a 
new test for non-Markovian effects. The second choice also allows for an 
estimate of the unitarity metric introduced in Ref.~\cite{Wallman2015} with no 
additional experimental overhead. In addition, we demonstrate that our protocol 
produces reliable estimates of loss rates through a numerical simulations under 
an error model that has the greatest variation in loss over states. 
Finally we illustrate how the analytical model breaks down when applied to 
systems that have reversible (coherent) leakage to an ancillary level.

\textit{Average survival rates}---In order to distinguish between inefficient 
detectors and lossy processes, we now define survival rates. Many methods for 
characterizing a process $\mathcal{E}$ (including randomized 
benchmarking~\cite{Emerson2005,Knill2008,Magesan2011}) assume it is 
trace-preserving. 
However, many experimental processes are not trace-preserving, but instead a 
state $\rho$ has a survival rate under $\mathcal{E}$
\begin{align}\label{def:survival_rate}
S(\rho|\mathcal{E}) = \frac{\tr[\mathcal{E}(\rho)]}{\tr\rho}
\end{align}
that is less than 1, or, equivalently, a nonzero loss rate 
$L(\rho|\mathcal{E})=1-S(\rho|\mathcal{E})$. Since the trace is linear and any 
unnormalized density matrix is proportional to a unit-trace density matrix, the 
survival rate averaged over all states (hereafter the average survival rate) is 
simply the survival rate of the maximally mixed state, that is, 
$S(\mathcal{E}):= S(\tfrac{1}{d}\mathds{I}|\mathcal{E})$. Correspondingly, the 
average loss rate is $L(\mathcal{E})=1-S(\mathcal{E})$.

\textit{Experimental protocol}---We now present a protocol for characterizing 
the average survival rate $S(\mathcal{E})$ in the experimental implementations 
$\{\mathcal{E}_g\}$ of a set of gates 
$\mathcal{G}=\{g_1,\ldots,g_{\av{\mathcal{G}}}\}$ that are at least a unitary 
1-design (e.g., the Pauli or Clifford groups)~\cite{Dankert2009}. For 
simplicity, we assume the noise is time- and gate-independent Markovian noise, 
so that $\mathcal{E}_g = g\circ\mathcal{E}$ for some fixed map $\mathcal{E}$ 
where $\circ$ denotes composition (i.e., apply $\mathcal{E}$ then $g$). This 
approach can be extended to accommodate time- and gate-dependent noise and a 
model of non-Markovian noise by applying the approaches of 
Refs.~\cite{Magesan2012a,Wallman2014,Ball2015}.

Our protocol for estimating $S(\mathcal{E})$ is as follows.
\begin{enumerate}
	\item Choose a sequence length $m\in\mathbb{N}$.
	\item Choose a random sequence ${\bf k} = (k_1,\ldots, k_m)$ of $m$ 
	integers uniformly at random, where $k_j\in \{1,\ldots,\av{\mathcal{G}}\}$.
	\item Prepare a state $\rho$.
	\item Apply the sequence of gates $g_{k_m}\circ\ldots\circ g_{k_1}$.
	\item Measure some operator $Q$ (e.g., a self-adjoint operator or POVM 
	element).
	\item Repeat steps 3--5 to estimate
	\begin{align}
		Q_{\bf k} = \tr\left[Q g_{k_m}\circ\mathcal{E}\circ\ldots\circ 		
		g_{k_1}\circ\mathcal{E}(\rho)\right]
	\end{align}
	to a desired precision.
	\item Repeat steps 2--6 to estimate
	\begin{align}
	\mathds{E}_{\bf k}Q_{\bf k} = \av{\mathcal{G}}^{-m}\sum_{\bf k} Q_{\bf k}
	\end{align}
	to a desired precision (see, e.g., Ref.~\cite{Wallman2014} for methods to bound
	the number of sequences required to obtain a given precision).
	\item Repeat steps 1--7 for multiple $m$ and fit to the decay curve
	\begin{align}\label{eq:decay_curve}
	\mathds{E}_{\bf k}Q_{\bf k} = D(Q)S(\rho|\mathcal{E})S^{m-1}(\mathcal{E}), 
	\end{align}
	derived below, to obtain estimates of $S(\mathcal{E})$ and 
	$S(\rho|\mathcal{E})D(Q)$ where $D(Q)=\tr\,Q/d$.
\end{enumerate}
(Note that the above protocol differs from the randomized benchmarking protocol 
of Ref.~\cite{Magesan2011} in that no inversion gate is applied prior to 
the measurement.)

Results of numerical simulations of our 
protocol for a specific loss model are illustrated in 
Fig.~\ref{fig:detectorerr_and_loss}, demonstrating the robust performance of 
our protocol. 

For the numerical simulation, the set of operations $\mathcal{G}$ is the set of 
single-qubit Paulis, and we modeled the error as
$\mathcal{E}$ as
\begin{align}
\mathcal{E}(\rho) = 
(\ket{0}\bra{0}+\alpha\ket{1}\bra{1})\rho(\ket{0}\bra{0}+\alpha\ket{1}\bra{1}),
\end{align}
where $\alpha=0.99$. The channel $\mathcal{E}$ corresponds to loss from the 
$\ket{1}$ state and, as proven in Proposition~\ref{prop:worst_case} below, has 
the greatest variation of loss over states. The measurement was set to $0.87 
\ket{\phi}\bra{\phi} + 0.95 \ket{\phi^{\bot}}\bra{\phi^{\bot}}$ where 
$\{\ket{\phi},\ket{\phi^{\bot}}\}$ is a randomly-chosen orthonormal basis to 
model an imperfect detector.

\begin{figure}[t]
	\centering
	\includegraphics[width=\linewidth,keepaspectratio]{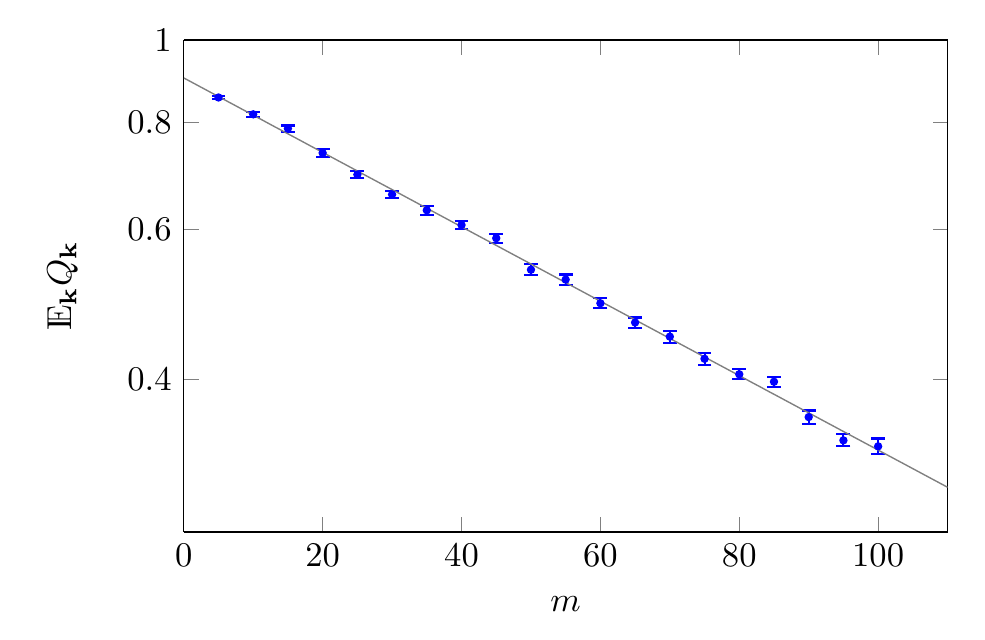}
	\caption{(Color online) Semilog plot of numerical data for our protocol demonstrating robust identification of  the average loss rate. The numerical data is obtained for  
		loss model described by Eq.~\eqref{eq:loss_channel}. 
The data points are the estimates of 
		$\mathds{E}_{\bf k}Q_{\bf k}$ for $m = 5, 10, ..., 100$ obtained by 
		sampling 30 random sequences of single-qubit Pauli operators (unitary 1-design) and the 
		error bars are the standard errors of the mean. The grey line is the 
		fit to the model in Eq.~\eqref{eq:decay_curve}, obtained using MATLAB's 
		nlinfit package, which gave $S(\mathcal{E}) = 0.9900(2)$ and $D(Q)= 
		0.902(8)$, compared to the theoretical values $S(\mathcal{E}) = 
		0.9901$ and $D(Q) = 0.910$ respectively.}
	\label{fig:detectorerr_and_loss}
\end{figure}

\textit{Analysis.}---To derive the decay curve in Eq.~\eqref{eq:decay_curve}, 
note that averaging over all sequences corresponds to independently averaging 
over all $g_{k_j}$, so that
\begin{align}\label{eq:average_operator}
	\mathds{E}_{\bf k}Q_{\bf k} &= \tr\left[Q 
	\overline{\mathcal{G}}\circ\mathcal{E}\circ\ldots\circ 		
	\overline{\mathcal{G}}\circ\mathcal{E}(\rho)\right]	\,
\end{align}
where $\overline{\mathcal{G}}(\rho) = 
\av{\mathcal{G}}^{-1}\sum_{g\in\mathcal{G}} g\rho g\ct$ (noting that a unitary 
channel corresponds to unitary conjugation). 
Since $\mathcal{G}$ is a unitary 1-design (and a linear map), 
$\overline{\mathcal{G}}(A) = \tr(A)\mathds{I}/d$ for all 
$d\times d$ matrices $A$~\cite{Dankert2009,Ambainis2000}. Therefore 
$\overline{\mathcal{G}}\circ\mathcal{E}(\rho) = 
S(\rho|\mathcal{E})\mathds{I}/d$ and 
$\overline{\mathcal{G}}\circ\mathcal{E}(\mathds{I}/d) = 
S(\mathcal{E})\mathds{I}/d$ and so Eq.~\eqref{eq:average_operator} simplifies 
to Eq.~\eqref{eq:decay_curve}.

The average survival rate obtained via our protocol is one possible figure of 
merit that could be used to characterize loss, with an important alternative 
being the worst-case loss. However, as we now prove, the average loss provides a 
practical bound for the worst-case loss:

\begin{prop}\label{prop:worst_case}
	For any quantum channel $\mathcal{E}$ and state $\rho$ for a 
	$d$-dimensional system,
	\begin{align*}
	L(\rho|\mathcal{E}) \leq L(\mathcal{E})d.
	\end{align*}
	Moreover, for all $d$ there exist channels $\mathcal{E}$ and states $\rho$ 
	that saturate this bound.
\end{prop}

\begin{proof}
	Let $\rho$ and $\mathcal{E}$ be arbitrary states of and channels for a 
	$d$-dimensional system. Since the trace is linear and any valid 
	state can be written as $\rho = \tau\tr \rho$ where $\tau$ is a unit-trace 
	density matrix, the survival rate is independent of $\tr \rho$, so we 
	assume 
	$\tr\rho=1$ without loss of generality.
	
	Let $\rho' = (\mathds{I}-\rho)/(d-1)$, which is a valid quantum state since 
	it 
	is Hermitian and positive-semidefinite by construction and has unit trace. 
	Since $\rho'$ is a valid quantum state, the probability of detecting a 
	system 
	in the state $\rho'$ after applying $\mathcal{E}$ is a true probability and 
	thus
	\begin{align}
	\frac{dS(\mathcal{E})-S(\rho|\mathcal{E})}{d-1}=\frac{\tr[\mathcal{E}(\mathds{I})-\mathcal{E}(\rho)]}{d-1}
	= \tr \mathcal{E}(\rho') \leq 1,
	\end{align}
	where we have used the fact that quantum channels and the trace are linear. 
	Rearranging and substituting $L = 1-S$ gives the desired bound.
	
	To see that the bound is saturated, fix $d$ and consider the channel
	\begin{align}\label{eq:loss_channel}
	\mathcal{E}(\rho) = [\mathds{I}+(\alpha-1) 
	\ket{0}\bra{0}]\rho[\mathds{I}+(\alpha-1) \ket{0}\bra{0}]
	\end{align} 
	for $\alpha\in[0,1]$. For this channel, 
	\begin{align}
	\mathcal{E}(\ket{j}\bra{j}) = \begin{cases} \alpha^2\ket{0}\bra{0} & j=0 \\
	\ket{j}\bra{j} & j\neq 0, \end{cases}
	\end{align} 
	so $L(\ket{j}\bra{j}|\mathcal{E})=\delta_j(1-\alpha^2)$ and 
	\begin{align}
	L(\mathcal{E}) = \frac{1}{d}\sum_j L(\ket{j}\bra{j}|\mathcal{E}) = 
	\frac{1-\alpha^2}{d}.
	\end{align} 
	Therefore there exists a channel $\mathcal{E}$ and a state $\rho$ such that 
	$L(\rho|\mathcal{E})=L(\mathcal{E})d$.
\end{proof}

For average survival rates close to 1, the estimate of $S(\rho|\mathcal{E})D(Q)$
can be used to directly estimate $D(Q)$, since
\begin{align}
S(\rho|\mathcal{E})D(Q) \in [(1-d [1-S(\mathcal{E})])D(Q),D(Q)]
\end{align}
by Proposition~\ref{prop:worst_case}. Consequently, 
$S(\rho|\mathcal{E})D(Q)/S(\mathcal{E})$ will give an estimate of $D(Q)$ that 
is accurate to within a factor of $(d-1)L(\mathcal{E})$. Estimating $D(Q)$ can 
be used to estimate the efficiency of the detector as
\begin{align}
\eta = \frac{D(Q)}{D(Q_{\rm ideal})},
\end{align}
where $Q_{\rm ideal}$ and $Q$ are the ideal and actual measurement operators. 
That is, $\eta$ is the ratio of observed to expected detector ``clicks'', 
averaged (independently) over all states.

\textit{Choosing states and measurements.}---Our protocol is robust to SPAM 
errors, in that the choice of the state $\rho$ and measurement operator $Q$ 
only effect the value of the constant $S(\rho|\mathcal{E})D(Q)$. However, there 
are two choices of $Q$ and $\rho$ that have particular benefits.

(i) The most useful scenario corresponds to choosing $\mathcal{G}$ to be a unitary 
2-design~\cite{Dankert2009} and choosing $\rho,Q\approx\ket{0}\bra{0}$ as in 
randomized benchmarking~\cite{Magesan2011}. There are two major advantages to 
this choice. Firstly, with this choice the same data can also be used to 
estimate the unitarity of $\mathcal{E}$, a quantitative measure of how the 
noise $\mathcal{E}$ effects the purity of input states~\cite{Wallman2015}. 
Secondly, estimating the constant prefactor in Eq.~\eqref{eq:decay_curve} with 
this choice is particularly useful because it allows an additional and vital 
constraint to be imposed when fitting randomized benchmarking data to the 
fidelity decay curve. In Ref.~\cite{Magesan2011}, it was shown that the 
fidelity decay curve is
\begin{align}
A(\mathcal{E}') p^m + B(\mathcal{E}') 
\end{align}
where $p$ is related to the average gate fidelity, $\mathcal{E}'$ is the 
average error under the convention that the experimental implementation of $g$ 
is written as $\mathcal{E}_g=\mathcal{E}'\circ g$ (in contrast to our choice of 
$\mathcal{E}_g=g\circ \mathcal{E}$) and
\begin{align}
A(\mathcal{E}') &= \tr[Q\mathcal{E}'(\rho-\tfrac{\mathds{I}}{d})] \notag\\
B(\mathcal{E}') &= \tr[Q\mathcal{E}'(\tfrac{\mathds{I}}{d})].
\end{align}
If the alternative convention of writing errors as $\mathcal{E}'\circ g$ is 
applied to Eq.~\eqref{eq:average_operator}, then the constant 
prefactor $S(\rho|\mathcal{E})D(Q)$ in Eq.~\eqref{eq:decay_curve} becomes 
$B(\mathcal{E}')$. Since the fidelity decay curve is in terms of observable 
properties, it is independent of the choice of convention and so 
$B(\mathcal{E}')=S(\rho|\mathcal{E})D(Q)$. Obtaining a precise estimate of the constant term for randomized benchmarking  is important for two reasons. First, underestimating the constant term 
$B(\mathcal{E}')$ [and hence overestimating the coefficient $A(\mathcal{E}')$] 
will lead to an overestimate of the decay parameter $p$, or, equivalently, an 
underestimate of the average gate infidelity. That is, \textit{underestimating 
the constant term will falsely indicate that the gates are performing better 
than they actually are}. Second, the values of the constants $A$ and $B$ in randomized 
benchmarking are not completely independent: they must satisfy particular 
constraints in order to correspond to physical Markovian noise processes. In 
particular, for qubits, note that
\begin{align}
B(\mathcal{E}) - A(\mathcal{E}) = \tr[Q\mathcal{E}(\rho^{\bot})]
\end{align}
where $\rho^{\bot}$ is the state whose Bloch vector is anti-parallel to that of 
$\rho$. Therefore $B(\mathcal{E}) - A(\mathcal{E})$ is a probability and so 
must be nonnegative if the noise is truly Markovian. Consequently, if 
$B(\mathcal{E}) - A(\mathcal{E})$ is (strongly) negative, then either the noise 
is non-Markovian or strongly gate dependent and so the estimate of the average 
gate infidelity in randomized benchmarking is not known to be accurate. 
Moreover, if the prefactor $S(\rho|\mathcal{E})D(Q)$ in 
Eq.~\eqref{eq:decay_curve} is estimated by setting $m=1$, then the resulting estimate is unaffected by the presence of non-Markovian effects  between sequential operations (since there is only one operation applied). Therefore if 
the estimate obtained by setting $m=1$ differs from the estimate obtained from 
fitting the randomized benchmarking data under the protocol of 
Ref.~\cite{Magesan2011}, then this disagreement indicates that non-Markovian effects are present in the data for the latter.

(ii) Alternatively, given any allowed choice of  $\mathcal{G}$, choosing $Q\approx \mathds{I}$ and $\rho$ to be any unit-trace density matrix will maximize the value of the constant prefactor in 
Eq.~\eqref{eq:decay_curve}, reducing the number of experiments required to 
obtain a desired precision  (since $\mathds{E}_{\bf k}Q_{\bf k}$ is close to one  for sufficiently small $m$). Note that this data can be collected under the same experimental configuration as case (i),  where 
$Q=\ket{0}\bra{0}$ and $\mathcal{G}$ is a unitary 2-design, by simply re-incorporating the outcomes associated with  $\mathds{I}-Q$ that are discarded in case (i). This data gives independent information because by assumption the probabilities of these two outcomes are not constrained to add to 1 due to presence of loss.

\textit{Coherent leakage.}---A distinct, but closely related error to loss is 
(coherent) leakage, wherein the system is ``leaked'' from the qubit subspace to 
other energy levels. Leakage errors are non-Markovian errors on the qubit 
subspace, since the system can return to the qubit subspace. Coherent leakage 
is a known consequence of control imperfections in certain implementations of 
the coupling gate in ion traps~\cite{Schindler2013} and the controlled-phase 
gate in superconducting qubits~\cite{DiCarlo2009,Dicarlo2010}. 
Fig.~\ref{fig:loss_failure} shows the results of our protocol given a model of 
coherent leakage, in particular, an error model for a random (fixed) unitary 
acting on a qutrit with a random relative phase between the leakage level and 
the qubit levels. The results initially appear to fit a single exponential 
decay, but then quickly converge to a constant, similar to the behavior 
observed in Ref.~\cite{Epstein2014}. Consequently, if experimental data for our 
protocol does not neatly fit a single exponential, one explanation would be 
that there is a leakage level that has not been accounted for. A simple 
protocol for estimating rates of coherent leakage has been provided previously 
in Ref.~\cite{Wallman2015a}.

\begin{figure}[t]
	\centering
	\includegraphics[width=\linewidth,keepaspectratio]{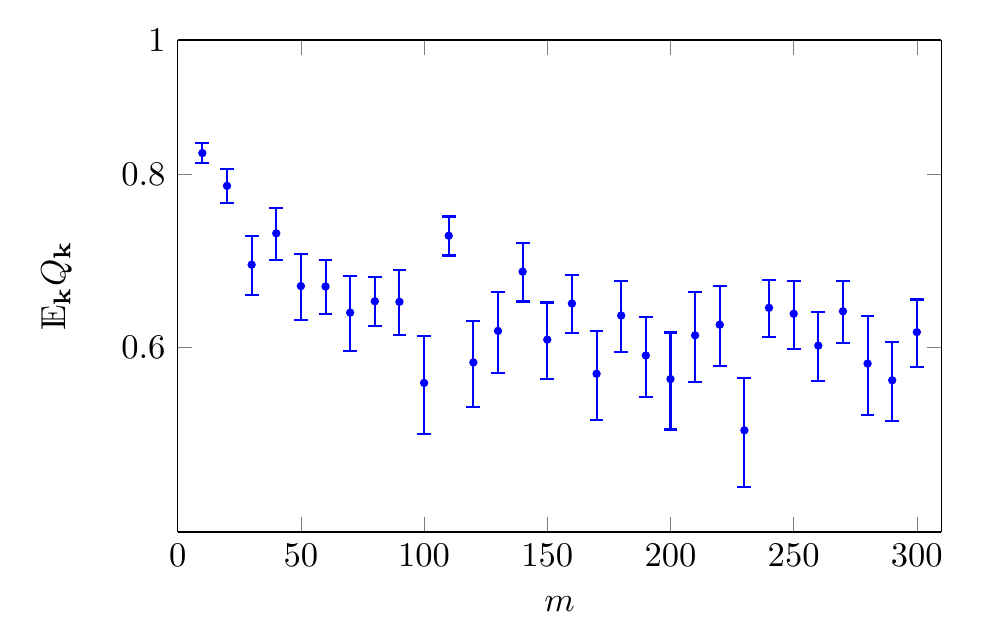}
	\caption{(Color online) Semi-log plot demonstrating the signature of non-Markovian leakage under our protocol. Numerical results obtained for a model of coherent leakage from a qubit subspace to a 
	third level under a small random unitary on the full qutrit space. The 
	data points are the estimates of $\mathds{E}_{\bf k}Q_{\bf k}$ for $m = 10, 
	20, ..., 300$ obtained by sampling 30 random sequences of single-qubit 
	Pauli operators and the error bars are the standard errors of the mean. }
	\label{fig:loss_failure}
\end{figure}

\textit{Conclusion}---In this paper, we have presented a platform-independent 
and robust protocol for characterizing the average loss rate due to noisy 
implementations of operations. Our protocol can also be used to estimate the
detector efficiency, provided the loss rate due to noisy operations is 
sufficiently small. Since our protocol is easy to implement, it is also a 
promising technique for experimentally optimizing quantum control, as done, 
e.g. in Ref.~\cite{Kelly2014} using randomized benchmarking experiments.

Experimentally implementing our protocol yields a single exponential decay 
curve which can be fitted to our analytical expression to obtain the average 
loss rate. If the experimental data deviates significantly from a single decay 
curve, the experimental noise is either strongly gate-dependent or 
non-Markovian. We have illustrated that the decay can be observed and 
fitted in practice through numerical simulations of loss for a specific error 
model and also that non-Markovian leakage to an ancillary level results in a 
deviation from a single exponential. However, fully characterizing how the 
present protocol (and other randomization-based protocols) behave in the 
presence of non-Markovian noise remains an open problem.

Our protocol is scalable and robust against state-preparation and measurement 
errors. However, particular choices of preparations and measurements give extra 
information. If the set of gates is chosen to be a unitary 2-design and the 
preparation and measurement are the same as those used in standard randomized 
benchmarking, then our current protocol can be applied to directly estimate one 
of the parameters in randomized benchmarking and thus provides a test to 
indicate non-Markovian noise. Furthermore, with this choice of preparation and 
measurement, the same data obtained via our protocol can be used to estimate 
the unitarity presented in Ref.~\cite{Wallman2015} and thus to estimate how 
close the noise is to depolarizing noise.

As with standard RB, obtaining rigorous confidence intervals on the parameters 
obtained from our protocol is still an open problem, though techniques bounding 
the number of sequences to be sampled~\cite{Wallman2014} and using Bayesian 
methods to refine prior information~\cite{Granade2014} should also be 
applicable to our protocol.

\textit{Acknowledgments}---The authors acknowledge helpful discussions with 
S.~Flammia, C.~Granade and T.~Monz. This research was supported by the U.S. 
Army Research Office through grant W911NF-14-1-0103, CIFAR, the Government of 
Ontario, and the Government of Canada through NSERC and Industry Canada.


\begin{thebibliography}{26}%
	\makeatletter
	\providecommand \@ifxundefined [1]{%
		\@ifx{#1\undefined}
	}%
	\providecommand \@ifnum [1]{%
		\ifnum #1\expandafter \@firstoftwo
		\else \expandafter \@secondoftwo
		\fi
	}%
	\providecommand \@ifx [1]{%
		\ifx #1\expandafter \@firstoftwo
		\else \expandafter \@secondoftwo
		\fi
	}%
	\providecommand \natexlab [1]{#1}%
	\providecommand \enquote  [1]{``#1''}%
	\providecommand \bibnamefont  [1]{#1}%
	\providecommand \bibfnamefont [1]{#1}%
	\providecommand \citenamefont [1]{#1}%
	\providecommand \href@noop [0]{\@secondoftwo}%
	\providecommand \href [0]{\begingroup \@sanitize@url \@href}%
	\providecommand \@href[1]{\@@startlink{#1}\@@href}%
	\providecommand \@@href[1]{\endgroup#1\@@endlink}%
	\providecommand \@sanitize@url [0]{\catcode `\\12\catcode `\$12\catcode
		`\&12\catcode `\#12\catcode `\^12\catcode `\_12\catcode `\%12\relax}%
	\providecommand \@@startlink[1]{}%
	\providecommand \@@endlink[0]{}%
	\providecommand \url  [0]{\begingroup\@sanitize@url \@url }%
	\providecommand \@url [1]{\endgroup\@href {#1}{\urlprefix }}%
	\providecommand \urlprefix  [0]{URL }%
	\providecommand \Eprint [0]{\href }%
	\providecommand \doibase [0]{http://dx.doi.org/}%
	\providecommand \selectlanguage [0]{\@gobble}%
	\providecommand \bibinfo  [0]{\@secondoftwo}%
	\providecommand \bibfield  [0]{\@secondoftwo}%
	\providecommand \translation [1]{[#1]}%
	\providecommand \BibitemOpen [0]{}%
	\providecommand \bibitemStop [0]{}%
	\providecommand \bibitemNoStop [0]{.\EOS\space}%
	\providecommand \EOS [0]{\spacefactor3000\relax}%
	\providecommand \BibitemShut  [1]{\csname bibitem#1\endcsname}%
	\let\auto@bib@innerbib\@empty
	\bibitem [{\citenamefont {Blakestad}\ \emph {et~al.}(2011)\citenamefont
		{Blakestad}, \citenamefont {Ospelkaus}, \citenamefont {Vandevender},
		\citenamefont {Wesenberg}, \citenamefont {Biercuk}, \citenamefont
		{Leibfried},\ and\ \citenamefont {Wineland}}]{Blakestad2011}%
	\BibitemOpen
	\bibfield  {author} {\bibinfo {author} {\bibfnamefont {R.~B.}\ \bibnamefont
			{Blakestad}}, \bibinfo {author} {\bibfnamefont {C.}~\bibnamefont
			{Ospelkaus}}, \bibinfo {author} {\bibfnamefont {A.~P.}\ \bibnamefont
			{Vandevender}}, \bibinfo {author} {\bibfnamefont {J.~H.}\ 
			\bibnamefont
			{Wesenberg}}, \bibinfo {author} {\bibfnamefont {M.~J.}\ \bibnamefont
			{Biercuk}}, \bibinfo {author} {\bibfnamefont {D.}~\bibnamefont 
			{Leibfried}},
		\ and\ \bibinfo {author} {\bibfnamefont {D.~J.}\ \bibnamefont 
		{Wineland}},\
	}\href {\doibase 10.1103/PhysRevA.84.032314} {\bibfield  {journal} {\bibinfo
		{journal} {Phys. Rev. A}\ }\textbf {\bibinfo {volume} {84}},\ \bibinfo
	{pages} {032314} (\bibinfo {year} {2011})},\ \Eprint
{http://arxiv.org/abs/1106.5005} {arXiv:1106.5005} \BibitemShut {NoStop}%
\bibitem [{\citenamefont {Wright}\ \emph {et~al.}(2013)\citenamefont {Wright},
	\citenamefont {Amini}, \citenamefont {Faircloth}, \citenamefont {Volin},
	\citenamefont {Doret}, \citenamefont {Hayden}, \citenamefont {Pai},
	\citenamefont {Landgren}, \citenamefont {Denison}, \citenamefont {Killian},
	\citenamefont {Slusher},\ and\ \citenamefont {Harter}}]{Wright2013}%
\BibitemOpen
\bibfield  {author} {\bibinfo {author} {\bibfnamefont {K.}~\bibnamefont
		{Wright}}, \bibinfo {author} {\bibfnamefont {J.~M.}\ \bibnamefont 
		{Amini}},
	\bibinfo {author} {\bibfnamefont {D.~L.}\ \bibnamefont {Faircloth}}, 
	\bibinfo
	{author} {\bibfnamefont {C.}~\bibnamefont {Volin}}, \bibinfo {author}
	{\bibfnamefont {S.~C.}\ \bibnamefont {Doret}}, \bibinfo {author}
	{\bibfnamefont {H.}~\bibnamefont {Hayden}}, \bibinfo {author} {\bibfnamefont
		{C.-S.}\ \bibnamefont {Pai}}, \bibinfo {author} {\bibfnamefont {D.~W.}\
		\bibnamefont {Landgren}}, \bibinfo {author} {\bibfnamefont 
		{D.}~\bibnamefont
		{Denison}}, \bibinfo {author} {\bibfnamefont {T.}~\bibnamefont 
		{Killian}},
	\bibinfo {author} {\bibfnamefont {R.~E.}\ \bibnamefont {Slusher}}, \ and\
	\bibinfo {author} {\bibfnamefont {A.~W.}\ \bibnamefont {Harter}},\ }\href
{\doibase 10.1088/1367-2630/15/3/033004} {\bibfield  {journal} {\bibinfo
		{journal} {New J. Phys.}\ }\textbf {\bibinfo {volume} {15}},\ \bibinfo
	{pages} {033004} (\bibinfo {year} {2013})},\ \Eprint
{http://arxiv.org/abs/1210.3655} {arXiv:1210.3655} \BibitemShut {NoStop}%
\bibitem [{\citenamefont {Vala}\ \emph {et~al.}(2005)\citenamefont {Vala},
	\citenamefont {Whaley},\ and\ \citenamefont {Weiss}}]{Vala2005}%
\BibitemOpen
\bibfield  {author} {\bibinfo {author} {\bibfnamefont {J.}~\bibnamefont
		{Vala}}, \bibinfo {author} {\bibfnamefont {K.~B.}\ \bibnamefont 
		{Whaley}}, \
	and\ \bibinfo {author} {\bibfnamefont {D.~S.}\ \bibnamefont {Weiss}},\ 
	}\href
{\doibase 10.1103/PhysRevA.72.052318} {\bibfield  {journal} {\bibinfo
		{journal} {Phys. Rev. A}\ }\textbf {\bibinfo {volume} {72}},\ \bibinfo
	{pages} {052318} (\bibinfo {year} {2005})},\ \Eprint
{http://arxiv.org/abs/0510021} {arXiv:0510021 [quant-ph]} \BibitemShut
{NoStop}%
\bibitem [{\citenamefont {Fortescue}\ \emph {et~al.}(2014)\citenamefont
	{Fortescue}, \citenamefont {Nawaf},\ and\ \citenamefont
	{Byrd}}]{Fortescue2014}%
\BibitemOpen
\bibfield  {author} {\bibinfo {author} {\bibfnamefont {B.}~\bibnamefont
		{Fortescue}}, \bibinfo {author} {\bibfnamefont {S.}~\bibnamefont 
		{Nawaf}}, \
	and\ \bibinfo {author} {\bibfnamefont {M.}~\bibnamefont {Byrd}},\ }\href
{http://arxiv.org/abs/1405.1766} {\  (\bibinfo {year} {2014})},\ \Eprint
{http://arxiv.org/abs/1405.1766v1} {arXiv:1405.1766v1} \BibitemShut {NoStop}%
\bibitem [{\citenamefont {Varnava}\ \emph {et~al.}(2006)\citenamefont
	{Varnava}, \citenamefont {Browne},\ and\ \citenamefont
	{Rudolph}}]{Varnava2006}%
\BibitemOpen
\bibfield  {author} {\bibinfo {author} {\bibfnamefont {M.}~\bibnamefont
		{Varnava}}, \bibinfo {author} {\bibfnamefont {D.~E.}\ \bibnamefont 
		{Browne}},
	\ and\ \bibinfo {author} {\bibfnamefont {T.}~\bibnamefont {Rudolph}},\ 
	}\href
{\doibase 10.1103/PhysRevLett.97.120501} {\bibfield  {journal} {\bibinfo
		{journal} {Phys. Rev. Lett.}\ }\textbf {\bibinfo {volume} {97}},\ 
		\bibinfo
	{pages} {120501} (\bibinfo {year} {2006})},\ \Eprint
{http://arxiv.org/abs/0507036} {arXiv:0507036 [quant-ph]} \BibitemShut
{NoStop}%
\bibitem [{\citenamefont {Whiteside}\ and\ \citenamefont
	{Fowler}(2014)}]{Whiteside2014}%
\BibitemOpen
\bibfield  {author} {\bibinfo {author} {\bibfnamefont {A.~C.}\ \bibnamefont
		{Whiteside}}\ and\ \bibinfo {author} {\bibfnamefont {A.~G.}\ 
		\bibnamefont
		{Fowler}},\ }\href {\doibase 10.1103/PhysRevA.90.052316} {\bibfield
	{journal} {\bibinfo  {journal} {Phys. Rev. A}\ }\textbf {\bibinfo {volume}
		{90}},\ \bibinfo {pages} {052316} (\bibinfo {year} {2014})}\BibitemShut
{NoStop}%
\bibitem [{\citenamefont {Hogg}\ \emph {et~al.}(2014)\citenamefont {Hogg},
	\citenamefont {Berry},\ and\ \citenamefont {Lvovsky}}]{Hogg2014}%
\BibitemOpen
\bibfield  {author} {\bibinfo {author} {\bibfnamefont {D.}~\bibnamefont
		{Hogg}}, \bibinfo {author} {\bibfnamefont {D.~W.}\ \bibnamefont 
		{Berry}}, \
	and\ \bibinfo {author} {\bibfnamefont {A.~I.}\ \bibnamefont {Lvovsky}},\
}\href {\doibase 10.1103/PhysRevA.90.053846} {\bibfield  {journal} {\bibinfo
	{journal} {Phys. Rev. A}\ }\textbf {\bibinfo {volume} {90}},\ \bibinfo
{pages} {053846} (\bibinfo {year} {2014})},\ \Eprint
{http://arxiv.org/abs/1408.0257v1} {arXiv:1408.0257v1} \BibitemShut {NoStop}%
\bibitem [{\citenamefont {Chuang}\ and\ \citenamefont
	{Nielsen}(1997)}]{Chuang1997}%
\BibitemOpen
\bibfield  {author} {\bibinfo {author} {\bibfnamefont {I.~L.}\ \bibnamefont
		{Chuang}}\ and\ \bibinfo {author} {\bibfnamefont {M.~A.}\ \bibnamefont
		{Nielsen}},\ }\href {\doibase 10.1080/09500349708231894} {\bibfield
	{journal} {\bibinfo  {journal} {J. Mod. Opt.}\ }\textbf {\bibinfo {volume}
		{44}},\ \bibinfo {pages} {2455} (\bibinfo {year} {1997})}\BibitemShut
{NoStop}%
\bibitem [{\citenamefont {Poyatos}\ \emph {et~al.}(1997)\citenamefont
	{Poyatos}, \citenamefont {Cirac},\ and\ \citenamefont
	{Zoller}}]{Poyatos1997}%
\BibitemOpen
\bibfield  {author} {\bibinfo {author} {\bibfnamefont {J.~F.}~\bibnamefont
		{Poyatos}}, \bibinfo {author} {\bibfnamefont {J.~I.}\ \bibnamefont 
		{Cirac}},
	\ and\ \bibinfo {author} {\bibfnamefont {P.}~\bibnamefont {Zoller}},\ }\href
{\doibase 10.1103/PhysRevLett.78.390} {\bibfield  {journal} {\bibinfo
		{journal} {Phys. Rev. Lett.}\ }\textbf {\bibinfo {volume} {78}},\ 
		\bibinfo
	{pages} {390} (\bibinfo {year} {1997})}\BibitemShut {NoStop}%
\bibitem [{\citenamefont {Merkel}\ \emph {et~al.}(2013)\citenamefont {Merkel},
	\citenamefont {Gambetta}, \citenamefont {Smolin}, \citenamefont {Poletto},
	\citenamefont {C\'{o}rcoles}, \citenamefont {Johnson}, \citenamefont 
	{Ryan},\
	and\ \citenamefont {Steffen}}]{Merkel2013}%
\BibitemOpen
\bibfield  {author} {\bibinfo {author} {\bibfnamefont {S.~T.}\ \bibnamefont
		{Merkel}}, \bibinfo {author} {\bibfnamefont {J.~M.}\ \bibnamefont
		{Gambetta}}, \bibinfo {author} {\bibfnamefont {J.~A.}\ \bibnamefont
		{Smolin}}, \bibinfo {author} {\bibfnamefont {S.}~\bibnamefont 
		{Poletto}},
	\bibinfo {author} {\bibfnamefont {A.~D.}\ \bibnamefont {C\'{o}rcoles}},
	\bibinfo {author} {\bibfnamefont {B.~R.}\ \bibnamefont {Johnson}}, \bibinfo
	{author} {\bibfnamefont {C.~A.}\ \bibnamefont {Ryan}}, \ and\ \bibinfo
	{author} {\bibfnamefont {M.}~\bibnamefont {Steffen}},\ }\href {\doibase
	10.1103/PhysRevA.87.062119} {\bibfield  {journal} {\bibinfo  {journal} 
	{Phys.
			Rev. A}\ }\textbf {\bibinfo {volume} {87}},\ \bibinfo {pages} 
			{062119}
	(\bibinfo {year} {2013})}\BibitemShut {NoStop}%
\bibitem [{\citenamefont {Magesan}\ \emph {et~al.}(2011)\citenamefont
	{Magesan}, \citenamefont {Gambetta},\ and\ \citenamefont
	{Emerson}}]{Magesan2011}%
\BibitemOpen
\bibfield  {author} {\bibinfo {author} {\bibfnamefont {E.}~\bibnamefont
		{Magesan}}, \bibinfo {author} {\bibfnamefont {J.~M.}\ \bibnamefont
		{Gambetta}}, \ and\ \bibinfo {author} {\bibfnamefont {J.}~\bibnamefont
		{Emerson}},\ }\href {\doibase 10.1103/PhysRevLett.106.180504} {\bibfield
	{journal} {\bibinfo  {journal} {Phys. Rev. Lett.}\ }\textbf {\bibinfo
		{volume} {106}},\ \bibinfo {pages} {180504} (\bibinfo {year}
	{2011})}\BibitemShut {NoStop}%
\bibitem [{\citenamefont {Wallman}\ \emph
	{et~al.}(2015{\natexlab{a}})\citenamefont {Wallman}, \citenamefont 
	{Granade},
	\citenamefont {Harper},\ and\ \citenamefont {Flammia}}]{Wallman2015}%
\BibitemOpen
\bibfield  {author} {\bibinfo {author} {\bibfnamefont {J.~J.}\ \bibnamefont
		{Wallman}}, \bibinfo {author} {\bibfnamefont {C.}~\bibnamefont 
		{Granade}},
	\bibinfo {author} {\bibfnamefont {R.}~\bibnamefont {Harper}}, \ and\ 
	\bibinfo
	{author} {\bibfnamefont {S.~T.}\ \bibnamefont {Flammia}},\ }\href@noop {} {\
	(\bibinfo {year} {2015}{\natexlab{a}})},\ \Eprint
{http://arxiv.org/abs/1503.07865v1} {arXiv:1503.07865v1} \BibitemShut
{NoStop}%
\bibitem [{\citenamefont {Emerson}\ \emph {et~al.}(2005)\citenamefont
	{Emerson}, \citenamefont {Alicki},\ and\ \citenamefont
	{\.{Z}yczkowski}}]{Emerson2005}%
\BibitemOpen
\bibfield  {author} {\bibinfo {author} {\bibfnamefont {J.}~\bibnamefont
		{Emerson}}, \bibinfo {author} {\bibfnamefont {R.}~\bibnamefont 
		{Alicki}}, \
	and\ \bibinfo {author} {\bibfnamefont {K.}~\bibnamefont {\.{Z}yczkowski}},\
}\href {\doibase 10.1088/1464-4266/7/10/021} {\bibfield  {journal} {\bibinfo
	{journal} {J. Opt. B Quantum Semiclassical Opt.}\ }\textbf {\bibinfo 
	{volume}
	{7}},\ \bibinfo {pages} {S347} (\bibinfo {year} {2005})}\BibitemShut
{NoStop}%
\bibitem [{\citenamefont {Knill}\ \emph {et~al.}(2008)\citenamefont {Knill},
	\citenamefont {Leibfried}, \citenamefont {Reichle}, \citenamefont {Britton},
	\citenamefont {Blakestad}, \citenamefont {Jost}, \citenamefont {Langer},
	\citenamefont {Ozeri}, \citenamefont {Seidelin},\ and\ \citenamefont
	{Wineland}}]{Knill2008}%
\BibitemOpen
\bibfield  {author} {\bibinfo {author} {\bibfnamefont {E.}~\bibnamefont
		{Knill}}, \bibinfo {author} {\bibfnamefont {D.}~\bibnamefont 
		{Leibfried}},
	\bibinfo {author} {\bibfnamefont {R.}~\bibnamefont {Reichle}}, \bibinfo
	{author} {\bibfnamefont {J.}~\bibnamefont {Britton}}, \bibinfo {author}
	{\bibfnamefont {R.~B.}\ \bibnamefont {Blakestad}}, \bibinfo {author}
	{\bibfnamefont {J.~D.}\ \bibnamefont {Jost}}, \bibinfo {author}
	{\bibfnamefont {C.}~\bibnamefont {Langer}}, \bibinfo {author} {\bibfnamefont
		{R.}~\bibnamefont {Ozeri}}, \bibinfo {author} {\bibfnamefont
		{S.}~\bibnamefont {Seidelin}}, \ and\ \bibinfo {author} {\bibfnamefont
		{D.~J.}\ \bibnamefont {Wineland}},\ }\href {\doibase
	10.1103/PhysRevA.77.012307} {\bibfield  {journal} {\bibinfo  {journal} 
	{Phys.
			Rev. A}\ }\textbf {\bibinfo {volume} {77}},\ \bibinfo {pages} 
			{012307}
	(\bibinfo {year} {2008})}\BibitemShut {NoStop}%
\bibitem [{\citenamefont {Dankert}\ \emph {et~al.}(2009)\citenamefont
	{Dankert}, \citenamefont {Cleve}, \citenamefont {Emerson},\ and\
	\citenamefont {Livine}}]{Dankert2009}%
\BibitemOpen
\bibfield  {author} {\bibinfo {author} {\bibfnamefont {C.}~\bibnamefont
		{Dankert}}, \bibinfo {author} {\bibfnamefont {R.}~\bibnamefont {Cleve}},
	\bibinfo {author} {\bibfnamefont {J.}~\bibnamefont {Emerson}}, \ and\
	\bibinfo {author} {\bibfnamefont {E.}~\bibnamefont {Livine}},\ }\href
{\doibase 10.1103/PhysRevA.80.012304} {\bibfield  {journal} {\bibinfo
		{journal} {Phys. Rev. A}\ }\textbf {\bibinfo {volume} {80}},\ \bibinfo
	{pages} {012304} (\bibinfo {year} {2009})}\BibitemShut {NoStop}%
\bibitem [{\citenamefont {Magesan}\ \emph {et~al.}(2012)\citenamefont
	{Magesan}, \citenamefont {Gambetta},\ and\ \citenamefont
	{Emerson}}]{Magesan2012a}%
\BibitemOpen
\bibfield  {author} {\bibinfo {author} {\bibfnamefont {E.}~\bibnamefont
		{Magesan}}, \bibinfo {author} {\bibfnamefont {J.~M.}\ \bibnamefont
		{Gambetta}}, \ and\ \bibinfo {author} {\bibfnamefont {J.}~\bibnamefont
		{Emerson}},\ }\href {\doibase 10.1103/PhysRevA.85.042311} {\bibfield
	{journal} {\bibinfo  {journal} {Phys. Rev. A}\ }\textbf {\bibinfo {volume}
		{85}},\ \bibinfo {pages} {042311} (\bibinfo {year} {2012})}\BibitemShut
{NoStop}%
\bibitem [{\citenamefont {Wallman}\ and\ \citenamefont
	{Flammia}(2014)}]{Wallman2014}%
\BibitemOpen
\bibfield  {author} {\bibinfo {author} {\bibfnamefont {J.~J.}\ \bibnamefont
		{Wallman}}\ and\ \bibinfo {author} {\bibfnamefont {S.~T.}\ \bibnamefont
		{Flammia}},\ }\href {\doibase 10.1088/1367-2630/16/10/103032} {\bibfield
	{journal} {\bibinfo  {journal} {New J. Phys.}\ }\textbf {\bibinfo {volume}
		{16}},\ \bibinfo {pages} {103032} (\bibinfo {year} {2014})}\BibitemShut
{NoStop}%
\bibitem [{\citenamefont {Ball}\ \emph {et~al.}(2015)\citenamefont {Ball},
	\citenamefont {Stace}, \citenamefont {Flammia},\ and\ \citenamefont
	{Biercuk}}]{Ball2015}%
\BibitemOpen
\bibfield  {author} {\bibinfo {author} {\bibfnamefont {H.}~\bibnamefont
		{Ball}}, \bibinfo {author} {\bibfnamefont {T.~M.}\ \bibnamefont 
		{Stace}},
	\bibinfo {author} {\bibfnamefont {S.~T.}\ \bibnamefont {Flammia}}, \ and\
	\bibinfo {author} {\bibfnamefont {M.~J.}\ \bibnamefont {Biercuk}},\ }\href
{http://arxiv.org/pdf/1504.05307.pdf} {\  (\bibinfo {year} {2015})},\ \Eprint
{http://arxiv.org/abs/1504.05307v1} {arXiv:1504.05307v1} \BibitemShut
{NoStop}%
\bibitem [{\citenamefont {Ambainis}\ \emph {et~al.}(2000)\citenamefont
	{Ambainis}, \citenamefont {Mosca}, \citenamefont {Tapp},\ and\ \citenamefont
	{Wolf}}]{Ambainis2000}%
\BibitemOpen
\bibfield  {author} {\bibinfo {author} {\bibfnamefont {A.}~\bibnamefont
		{Ambainis}}, \bibinfo {author} {\bibfnamefont {M.}~\bibnamefont 
		{Mosca}},
	\bibinfo {author} {\bibfnamefont {A.}~\bibnamefont {Tapp}}, \ and\ \bibinfo
	{author} {\bibfnamefont {R.~D.}\ \bibnamefont {Wolf}},\ }\href {\doibase
	10.1109/SFCS.2000.892142} {\bibfield  {journal} {\bibinfo  {journal} {Proc.
			41st Annu. Symp. Found. Comput. Sci.}\ } (\bibinfo {year} {2000}),\
	10.1109/SFCS.2000.892142}\BibitemShut {NoStop}%
\bibitem [{\citenamefont {Schindler}\ \emph {et~al.}(2013)\citenamefont
	{Schindler}, \citenamefont {Nigg}, \citenamefont {Monz}, \citenamefont
	{Barreiro}, \citenamefont {Martinez}, \citenamefont {Wang}, \citenamefont
	{Quint}, \citenamefont {Brandl}, \citenamefont {Nebendahl}, \citenamefont
	{Roos}, \citenamefont {Chwalla}, \citenamefont {Hennrich},\ and\
	\citenamefont {Blatt}}]{Schindler2013}%
\BibitemOpen
\bibfield  {author} {\bibinfo {author} {\bibfnamefont {P.}~\bibnamefont
		{Schindler}}, \bibinfo {author} {\bibfnamefont {D.}~\bibnamefont 
		{Nigg}},
	\bibinfo {author} {\bibfnamefont {T.}~\bibnamefont {Monz}}, \bibinfo 
	{author}
	{\bibfnamefont {J.~T.}\ \bibnamefont {Barreiro}}, \bibinfo {author}
	{\bibfnamefont {E.}~\bibnamefont {Martinez}}, \bibinfo {author}
	{\bibfnamefont {S.~X.}\ \bibnamefont {Wang}}, \bibinfo {author}
	{\bibfnamefont {S.}~\bibnamefont {Quint}}, \bibinfo {author} {\bibfnamefont
		{M.~F.}\ \bibnamefont {Brandl}}, \bibinfo {author} {\bibfnamefont
		{V.}~\bibnamefont {Nebendahl}}, \bibinfo {author} {\bibfnamefont 
		{C.~F.}\
		\bibnamefont {Roos}}, \bibinfo {author} {\bibfnamefont {M.}~\bibnamefont
		{Chwalla}}, \bibinfo {author} {\bibfnamefont {M.}~\bibnamefont 
		{Hennrich}}, \
	and\ \bibinfo {author} {\bibfnamefont {R.}~\bibnamefont {Blatt}},\ }\href
{\doibase 10.1088/1367-2630/15/12/123012} {\bibfield  {journal} {\bibinfo
		{journal} {New J. Phys.}\ }\textbf {\bibinfo {volume} {15}} (\bibinfo 
		{year}
	{2013}),\ 10.1088/1367-2630/15/12/123012},\ \Eprint
{http://arxiv.org/abs/1308.3096} {arXiv:1308.3096} \BibitemShut {NoStop}%
\bibitem [{\citenamefont {DiCarlo}\ \emph {et~al.}(2009)\citenamefont
	{DiCarlo}, \citenamefont {Chow}, \citenamefont {Gambetta}, \citenamefont
	{Bishop}, \citenamefont {Johnson}, \citenamefont {Schuster}, \citenamefont
	{Majer}, \citenamefont {Blais}, \citenamefont {Frunzio}, \citenamefont
	{Girvin},\ and\ \citenamefont {Schoelkopf}}]{DiCarlo2009}%
\BibitemOpen
\bibfield  {author} {\bibinfo {author} {\bibfnamefont {L.}~\bibnamefont
		{DiCarlo}}, \bibinfo {author} {\bibfnamefont {J.~M.}\ \bibnamefont 
		{Chow}},
	\bibinfo {author} {\bibfnamefont {J.~M.}\ \bibnamefont {Gambetta}}, \bibinfo
	{author} {\bibfnamefont {L.~S.}\ \bibnamefont {Bishop}}, \bibinfo {author}
	{\bibfnamefont {B.~R.}\ \bibnamefont {Johnson}}, \bibinfo {author}
	{\bibfnamefont {D.~I.}\ \bibnamefont {Schuster}}, \bibinfo {author}
	{\bibfnamefont {J.}~\bibnamefont {Majer}}, \bibinfo {author} {\bibfnamefont
		{A.}~\bibnamefont {Blais}}, \bibinfo {author} {\bibfnamefont
		{L.}~\bibnamefont {Frunzio}}, \bibinfo {author} {\bibfnamefont {S.~M.}\
		\bibnamefont {Girvin}}, \ and\ \bibinfo {author} {\bibfnamefont {R.~J.}\
		\bibnamefont {Schoelkopf}},\ }\href {\doibase 10.1038/nature08121} 
		{\bibfield
	{journal} {\bibinfo  {journal} {Nature}\ }\textbf {\bibinfo {volume}
		{460}},\ \bibinfo {pages} {240} (\bibinfo {year} {2009})},\ \Eprint
{http://arxiv.org/abs/0903.2030} {arXiv:0903.2030} \BibitemShut {NoStop}%
\bibitem [{\citenamefont {DiCarlo}\ \emph {et~al.}(2010)\citenamefont
	{DiCarlo}, \citenamefont {Reed}, \citenamefont {Sun}, \citenamefont
	{Johnson}, \citenamefont {Chow}, \citenamefont {Gambetta}, \citenamefont
	{Frunzio}, \citenamefont {Girvin}, \citenamefont {Devoret},\ and\
	\citenamefont {Schoelkopf}}]{Dicarlo2010}%
\BibitemOpen
\bibfield  {author} {\bibinfo {author} {\bibfnamefont {L.}~\bibnamefont
		{DiCarlo}}, \bibinfo {author} {\bibfnamefont {M.~D.}\ \bibnamefont 
		{Reed}},
	\bibinfo {author} {\bibfnamefont {L.}~\bibnamefont {Sun}}, \bibinfo {author}
	{\bibfnamefont {B.~R.}\ \bibnamefont {Johnson}}, \bibinfo {author}
	{\bibfnamefont {J.~M.}\ \bibnamefont {Chow}}, \bibinfo {author}
	{\bibfnamefont {J.~M.}\ \bibnamefont {Gambetta}}, \bibinfo {author}
	{\bibfnamefont {L.}~\bibnamefont {Frunzio}}, \bibinfo {author} 
	{\bibfnamefont
		{S.~M.}\ \bibnamefont {Girvin}}, \bibinfo {author} {\bibfnamefont 
		{M.~H.}\
		\bibnamefont {Devoret}}, \ and\ \bibinfo {author} {\bibfnamefont 
		{R.~J.}\
		\bibnamefont {Schoelkopf}},\ }\href {\doibase 10.1038/nature09416} 
		{\bibfield
	{journal} {\bibinfo  {journal} {Nature}\ }\textbf {\bibinfo {volume}
		{467}},\ \bibinfo {pages} {574} (\bibinfo {year} {2010})},\ \Eprint
{http://arxiv.org/abs/1004.4324} {arXiv:1004.4324} \BibitemShut {NoStop}%
\bibitem [{\citenamefont {Epstein}\ \emph {et~al.}(2014)\citenamefont
	{Epstein}, \citenamefont {Cross}, \citenamefont {Magesan},\ and\
	\citenamefont {Gambetta}}]{Epstein2014}%
\BibitemOpen
\bibfield  {author} {\bibinfo {author} {\bibfnamefont {J.~M.}\ \bibnamefont
		{Epstein}}, \bibinfo {author} {\bibfnamefont {A.~W.}\ \bibnamefont 
		{Cross}},
	\bibinfo {author} {\bibfnamefont {E.}~\bibnamefont {Magesan}}, \ and\
	\bibinfo {author} {\bibfnamefont {J.~M.}\ \bibnamefont {Gambetta}},\ }\href
{\doibase 10.1103/PhysRevA.89.062321} {\bibfield  {journal} {\bibinfo
		{journal} {Phys. Rev. A}\ }\textbf {\bibinfo {volume} {89}},\ \bibinfo
	{pages} {062321} (\bibinfo {year} {2014})}\BibitemShut {NoStop}%
\bibitem [{\citenamefont {Wallman}\ \emph
	{et~al.}(2015{\natexlab{b}})\citenamefont {Wallman}, \citenamefont
	{Barnhill},\ and\ \citenamefont {Emerson}}]{Wallman2015a}%
\BibitemOpen
\bibfield  {author} {\bibinfo {author} {\bibfnamefont {J.~J.}\ \bibnamefont
		{Wallman}}, \bibinfo {author} {\bibfnamefont {M.}~\bibnamefont 
		{Barnhill}}, \
	and\ \bibinfo {author} {\bibfnamefont {J.}~\bibnamefont {Emerson}},\
}\href@noop {} {\  (\bibinfo {year} {2015}{\natexlab{b}})},\ \Eprint
{http://arxiv.org/abs/1412.4126v2} {arXiv:1412.4126v2} \BibitemShut {NoStop}%
\bibitem [{\citenamefont {Kelly}\ \emph {et~al.}(2014)\citenamefont {Kelly},
	\citenamefont {Barends}, \citenamefont {Campbell}, \citenamefont {Chen},
	\citenamefont {Chen}, \citenamefont {Chiaro}, \citenamefont {Dunsworth},
	\citenamefont {Fowler}, \citenamefont {Hoi}, \citenamefont {Jeffrey},
	\citenamefont {Megrant}, \citenamefont {Mutus}, \citenamefont {Neill},
	\citenamefont {O`Malley}, \citenamefont {Quintana}, \citenamefont {Roushan},
	\citenamefont {Sank}, \citenamefont {Vainsencher}, \citenamefont {Wenner},
	\citenamefont {White}, \citenamefont {Cleland},\ and\ \citenamefont
	{Martinis}}]{Kelly2014}%
\BibitemOpen
\bibfield  {author} {\bibinfo {author} {\bibfnamefont {J.}~\bibnamefont
		{Kelly}}, \bibinfo {author} {\bibfnamefont {R.}~\bibnamefont {Barends}},
	\bibinfo {author} {\bibfnamefont {B.}~\bibnamefont {Campbell}}, \bibinfo
	{author} {\bibfnamefont {Y.}~\bibnamefont {Chen}}, \bibinfo {author}
	{\bibfnamefont {Z.}~\bibnamefont {Chen}}, \bibinfo {author} {\bibfnamefont
		{B.}~\bibnamefont {Chiaro}}, \bibinfo {author} {\bibfnamefont
		{A.}~\bibnamefont {Dunsworth}}, \bibinfo {author} {\bibfnamefont 
		{A.~G.}\
		\bibnamefont {Fowler}}, \bibinfo {author} {\bibfnamefont {I.~C.}\
		\bibnamefont {Hoi}}, \bibinfo {author} {\bibfnamefont {E.}~\bibnamefont
		{Jeffrey}}, \bibinfo {author} {\bibfnamefont {A.}~\bibnamefont 
		{Megrant}},
	\bibinfo {author} {\bibfnamefont {J.}~\bibnamefont {Mutus}}, \bibinfo
	{author} {\bibfnamefont {C.}~\bibnamefont {Neill}}, \bibinfo {author}
	{\bibfnamefont {P.~J.~J.}\ \bibnamefont {O`Malley}}, \bibinfo {author}
	{\bibfnamefont {C.}~\bibnamefont {Quintana}}, \bibinfo {author}
	{\bibfnamefont {P.}~\bibnamefont {Roushan}}, \bibinfo {author} 
	{\bibfnamefont
		{D.}~\bibnamefont {Sank}}, \bibinfo {author} {\bibfnamefont 
		{A.}~\bibnamefont
		{Vainsencher}}, \bibinfo {author} {\bibfnamefont {J.}~\bibnamefont 
		{Wenner}},
	\bibinfo {author} {\bibfnamefont {T.~C.}\ \bibnamefont {White}}, \bibinfo
	{author} {\bibfnamefont {A.~N.}\ \bibnamefont {Cleland}}, \ and\ \bibinfo
	{author} {\bibfnamefont {J.~M.}\ \bibnamefont {Martinis}},\ }\href
{http://arxiv.org/abs/1403.0035} {\bibfield  {journal} {\bibinfo  {journal}
		{Phys. Rev. Lett.}\ }\textbf {\bibinfo {volume} {112}},\ \bibinfo 
		{pages}
	{240504} (\bibinfo {year} {2014})},\ \Eprint 
	{http://arxiv.org/abs/1403.0035}
{arXiv:1403.0035} \BibitemShut {NoStop}%
\bibitem [{\citenamefont {Granade}\ \emph {et~al.}(2014)\citenamefont
	{Granade}, \citenamefont {Ferrie},\ and\ \citenamefont 
	{Cory}}]{Granade2014}%
\BibitemOpen
\bibfield  {author} {\bibinfo {author} {\bibfnamefont {C.}~\bibnamefont
		{Granade}}, \bibinfo {author} {\bibfnamefont {C.}~\bibnamefont 
		{Ferrie}}, \
	and\ \bibinfo {author} {\bibfnamefont {D.~G.}\ \bibnamefont {Cory}},\ }\href
{\doibase 10.1088/1367-2630/17/1/013042} {\bibfield  {journal} {\bibinfo
		{journal} {New J. Phys.}\ }\textbf {\bibinfo {volume} {17}},\ \bibinfo
	{pages} {013042} (\bibinfo {year} {2014})},\ \Eprint
{http://arxiv.org/abs/1404.5275v1} {arXiv:1404.5275v1} \BibitemShut {NoStop}%
\end{thebibliography}
\end{document}